\documentclass[runningheads]{llncs}
\usepackage{graphicx}
\usepackage{amsmath, amssymb}
\usepackage[utf8]{inputenc}
\usepackage{algorithm}
\usepackage{algpseudocode}
\usepackage{enumerate}\usepackage{color}   
\usepackage{multirow}

\usepackage{verbatim}
\usepackage{mathtools}
\usepackage{float}
\usepackage{cite}
\begin{document}
\title{A Scheme to resist Fast Correlation Attack for Word Oriented LFSR based Stream Cipher}

\author{Subrata Nandi\inst{1} \and Srinivasan Krishnaswamy\inst{2}
\and Pinaki Mitra\inst{1}}

\institute{Department Of Computer Science and Engineering,Indian Institute of Technology Guwahati,India\\
\email{subrata.nandi@iitg.ac.in}\\
\email{pinaki@iitg.ac.in}
\and
Department of Electrical and Electronics and Electrical Engineering,Indian Institute of Technology Guwahati,India\\
\email{srinikris@iitg.ac.in}
}

\maketitle     
\begin{abstract}
In LFSR-based stream ciphers, the knowledge of the feedback equation of the LFSR plays a critical role in most attacks. In word-based stream ciphers such as those in the SNOW series, even if the feedback configuration is hidden, knowing the characteristic polynomial of the state transition matrix of the LFSR enables the attacker to create a feedback equation over $GF(2)$. This, in turn, can be used to launch fast correlation attacks.
In this work, we propose a  method for hiding both the feedback equation of a  word-based LFSR and the characteristic polynomial of the state transition matrix.  Here, we employ a $z$-primitive $\sigma$-LFSR whose characteristic polynomial is randomly sampled from the distribution of primitive polynomials over $GF(2)$ of the appropriate degree. We propose an algorithm for locating $z$-primitive $\sigma$-LFSR configurations of a given degree. Further, an invertible matrix is generated from the key. This is then employed to generate a public parameter which is used to retrieve the feedback configuration using the key. If the key size is  $n$- bits,  the process of retrieving  the feedback equation from the public parameter has a average time complexity $\mathbb{O}(2^{n-1})$.
The proposed method has been tested on SNOW 2.0 and SNOW 3G for resistance to fast correlation attacks. We have demonstrated that the security of SNOW 2.0 and SNOW 3G increases from 128 bits to 256 bits.

\end{abstract}

\section{Introduction}
Stream ciphers are often used to secure information sent over insecure communication channels. Here, the ciphertext is the XOR sum of a pseudo-random keystream and the plaintext. Many stream ciphers are designed using LFSRs as they are extremely easy to implement both in hardware and software. Various word-oriented stream cipher configurations have been proposed to effectively utilize the word based architecture of modern processors. These include SNOW 2.0\cite{ekdahl2002new}, SNOW 3G\cite{orhanou2010snow}, SNOW V\cite{ekdahl2019new}and Sosemanuk \cite{berbain2008sosemanuk}. These ciphers are based on LFSRs with multi-bit delay blocks and allow for extremely quick software implementations. 

The feedback configurations of these LFSRs are publicly known, and this plays a key role in many well-known plaintext attacks such as algebraic attacks, fast correlation attacks, distinguishing attacks, guess and determine attacks etc.
The KDFC scheme\cite{nandi2022key} provides resistance against some of these attacks by concealing the feedback configuration. However, in this scheme, the characteristic polynomial of the LFSR over $GF(2)$ is known. This, in turn, leads to cryptanalysis by a Fast correlation attack(FCA) \cite{gong2020fast}. In this article, we propose a scheme to resist such attacks.
 \subsection{Related Works}
 Multiple correlation attacks and distinguishing attacks have been proposed for SNOW 2.0 and SNOW 3G\cite{watanabe2003distinguishing,nyberg2006improved,maximov2005fast,lee2008cryptanalysis,zhang2015fast,funabiki2018several,yang2019vectorized,gong2020fast}. The first description of a Fast correlation attack for word-based stream ciphers is found in \cite{meier1989fast}. An enhanced version of this attack is given in \cite{chepyzhov2000simple}. These attacks utilize a linear recurring relation with coefficients in $F_2$.
 For SNOW 2.0, this relation has a degree of 512. The attack described in \cite{lee2008cryptanalysis} has a time complexity of $2^{212.38}$. The attack described in \cite{zhang2015fast} considers the LFSR in SNOW2.0 as being over $F_{2^8}$. This results in a linear recurring relation of degree 64. Here, parity check equations are generated using the Wagner's $k$-tree algorithm described in \cite{wagner2002generalized}. The time complexity of this attack is $2^{164.5}$, which is roughly $2^{49}$ times better than the attack described in \cite{lee2008cryptanalysis}. However, the feedback function of LFSR (as an equation over $F_{2^{32}}$) is critical for this attack. \cite{funabiki2018several} explains a Mixed Integer Linear Programming(MILP) based linear mask search on SNOW 2.0 to find a better correlation $2^{-14.411}$ for the FSM approximation equation. It recovers the Key of SNOW 2.0 with time complexity $2^{162.91}$.  \cite{yang2019vectorized} describes a vectorized linear approximation attack on SNOW 3G with a bias value of $2^{-40}$ and time complexity of  $2^{177}$ to find the state of the SNOW 3G LFSR. \cite{gong2020fast} proposes another attack model based on a modified Wagner K-tree algorithm and linear approximation of some composition function with time complexity $2^{162.86}$  for SNOW 2.0 as and time complexity $2^{222.33}$ for SNOW 3G. This attack considers the LFSR as a state transition machine over $\mathbb{F}_2$ and uses the characteristic polynomial of the state transition matrix.  This polynomial has degree 512. All these algorithms need a linear recurring relation that the LFSR satisfies. Our main contribution, as stated in the following subsection, is to deny the attacker the knowledge of any such linear recurring relation.
 \subsection{ Our contributions}
 In this paper, we give an algorithm to generate a $z-$ primitive $\sigma-$LFSR configuration with a random primitive characteristic polynomial which is sampled from the uniform distribution on all primitive polynomial of a given degree. In the proposed scheme, there are two levels of randomness; in the choice of the primitive characteristic polynomial of certain degree (over $GF(2)$) and in the choice of feedback configuration for a given characteristic polynomial. The purpose of doing this is to deny the attacker the knowledge of any linear recurring relation that the LFSR satisfies. This is done to counter Fast Correlation Attacks which use such linear recurring relations to generate parity check equations.    
 \subsection{Organization of the Paper}
 Section 2 deals with notations, definitions and theorems related to $z-$primitive $\sigma-$LFSRs and Fast Correlation Attacks. In section 3, we describe a $z-$primitive $\sigma-$LFSR generation algorithm with proof of correctness and a short example. Section 4 gives a brief description of stream ciphers SNOW 2.0 and SNOW 3G. Section 5 illustrates the proposed method for pre-processing, initialization and key generation. In section 6, this method is applied on SNOW 2.0 and SNOW 3G and the resulting resistance to fast correlation attacks is demonstrated . Finally, the conclusions along with some additional discussion and future work are outlined in Section 7.

\section{Preliminearies}
In this section,  we introduce  notations and definitions that are used throughout this article. 

\subsection{Notations}
The following is the list of notations used in this paper:
  \begin{table}[H]
  	\centering
  	\begin{tabular}{c|c}
  		\hline
  		\textbf{Symbol}   & \textbf{Meaning}\\
  		\hline
  		$GF(p^n)$,$\mathbb{F}_{p^n}$  &  Finite field of cardinality $p^n$, where p is a prime \\
  		\hline
    
  		$\mathbb{F}_p^n$ &   $n$-dimensional vector space over $\mathbb{F}_p$             \\
  		\hline
  		$GF(P)$ & Galois Field with elements $\in\{0,1,\cdots,p-1\}$\\
  		\hline
		$M_m(F_2)$ & Matrix Ring over $F_2$\\
  		  		\hline
            $\delta$ & Circular Left shift Operator of a Sequence \\
  		$\mathbb{N}$ & Natural Number\\
  		\hline
  		$\mathbb{R}$ & Real Number\\
  		\hline
  		$P(x)$ & Probability of a random varibale x\\
  		\hline
  		$v^T$ & Transpose of a vector $v$.\\
  		\hline
  		$v_1.v_2$ & Dot product between two vectors $v_1$ and $v_2$ over $F_2$\\
  		\hline
  		$\oplus $ &     XOR    \\  
  		\hline
  		$\boxplus$ & Addition modulo $2^n$\\
  		\hline
  		$M^{n \times n}$ & Matrix M with $n$ rows and $n$ columns\\
  		\hline
  		$e_1^i$ & row vector with $i-$th entry is 1 and the other entries are 0. \\
  		\hline
  		$A^T$ & Transpose of a matrix $A$.\\
  		\hline
  		$M^{-1}$ & Inverse of a matrix $M$. \\
  		\hline
  		GCD(a,b) & Greatest common divisiors of a and b\\
  		\hline
            $M[i,j]$ & $j-$th element of $i-$th row of $M$ matrix.\\
            \hline
            $M[i,:]$ & $i-$th row of a Matrix $M$\\
            \hline
            $M[:,i]$ & $i-$th column of a Matrix $M$\\
            \hline
            $[v_1,v_2,\ldots,v_n]$ & A matrix with columns $v_1,v_2,\ldots,v_n$\\
             \hline
            $[v_1;v_2;\ldots;v_n]$ & A matrix with rows $v_1,v_2,\ldots,v_n$\\
             \hline
            $\&\&$ & Bitwise AND Operator\\
            \hline
            $||$ & Concatenation Operator\\
            \hline
  	\end{tabular}
  	\caption{Symbol and their meaning}
  	\label{tab:my_label}
  \end{table}
  \subsection{SNOW 2.0 and SNOW 3G}
The SNOW series of stream ciphers represent a class of stream ciphers where a word based Linear Feedback Shift Register (LFSR) is coupled with a Finite State Machine(FSM). The keystream is the sum of the outputs of the the LFSR and the FSM
 The main distinction between SNOW 2.0 and SNOW 3G  is that the later has an additional register in the FSM (designated by $Rt3$).
In both schemes, the LFSR is composed of $16$ $32$-bit delay blocks with the feedback polynomial $\gamma .z^{16} + z^{14}+ \gamma^{-1}+z^5+1 \in F_{2^{32}}[z]$, where $\gamma \in  F_2^{32}$ is a root of the primitive polynomial $x^4+\alpha^{23} x^3+ \alpha^{245}.x^2+\alpha^{48}x+\alpha^{239} \in F_{2^{8}}[x]$, and is a root of the polynomial $y^8+y^7+y^5+y^3+1 \in F_2[y]$.
\begin{figure}
	\centering
	\includegraphics[width=1.1\linewidth, height=.4\textheight]{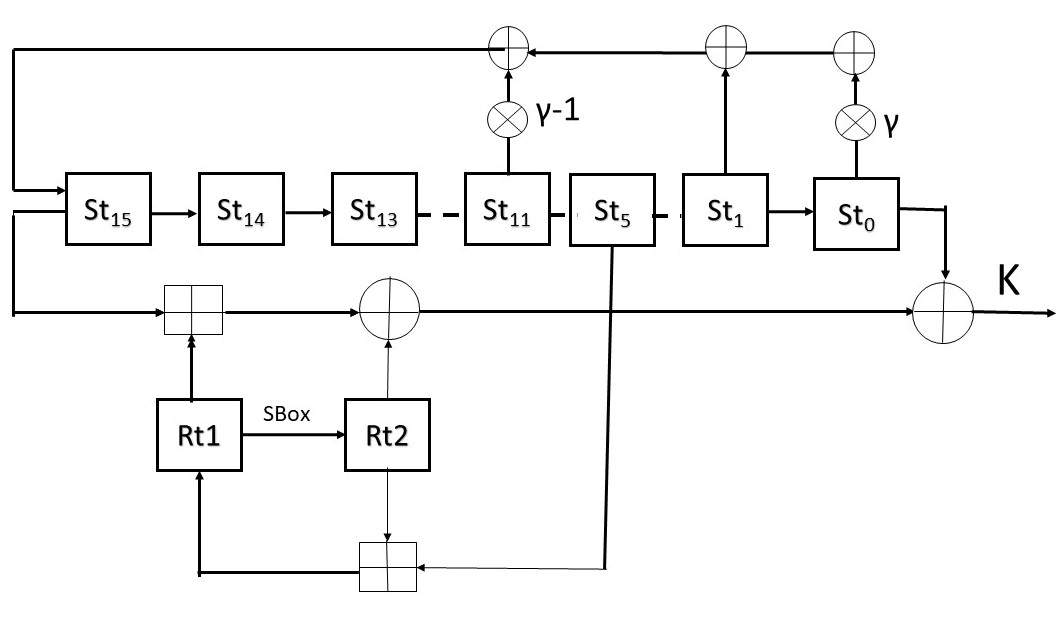}
	\caption{Block Diagram of SNOW 2.0}
	\label{fig:snow-2}
\end{figure}

At time instant  $t\ge 0$,  the state of the LFSR is defined as $(St_{t+15}, St_{t+14},..., St_t) \in F_2^{32\times 16}$. 

Let the output of the registers in the of FSM of SNOW2 at time instant $t$ be denoted by $Rt1_t$ and $Rt2_t$. The output of the FSM at time instant $t$ is denoted by $F_t$, The FSM in SNOW2 is governed by the following equations
\begin{eqnarray*}\label{eq1}
	F_t&=&(St_{t+15}\boxplus {Rt1}_t) \oplus {Rt2}_t, t\ge 0\\
 Rt1_{t+1}&=&St_{t+5} \boxplus Rt2_t\\
		Rt2_{t+1}&=&SBox(Rt1_t)
\end{eqnarray*}

where $Sbox$ is a bijection over $F_2^{32}$, composed of  four parallel AES S-boxes  followed by the AES MixColumn transformation.\par`

Let the outputs of the registers in the FSM of SNOW3G at time instant $t$ be denoted by
$Rt1_t, Rt2_t$ and $Rt3_t$ respectively. Let the output of the FSM at time instant $t$ be denoted by $F_t$. The FSM is governed by the following equations:
\begin{eqnarray*}
Ft_t &=& (St_{t+15} \boxplus Rt1_t)\oplus Rt2_t\\
Rt1_{t+1}&=&(St_{t+5} \oplus Rt3_t)\boxplus Rt2_t\\
		Rt2_{t+1}&=&SBox1(Rt1_t)\\
		Rt3_{t+1}&=&SBox2(Rt2_t)
\end{eqnarray*}
where $SBox1$ in equation \ref{eq1} is same as the SBox of SNOW 2.0 while $SBox2$ is  another bijection over $GF(2^{32})$,  based on the Dickson polynomial.
\begin{figure}[H]
	\centering
	\includegraphics[width=1.1\linewidth, height=0.4\textheight]{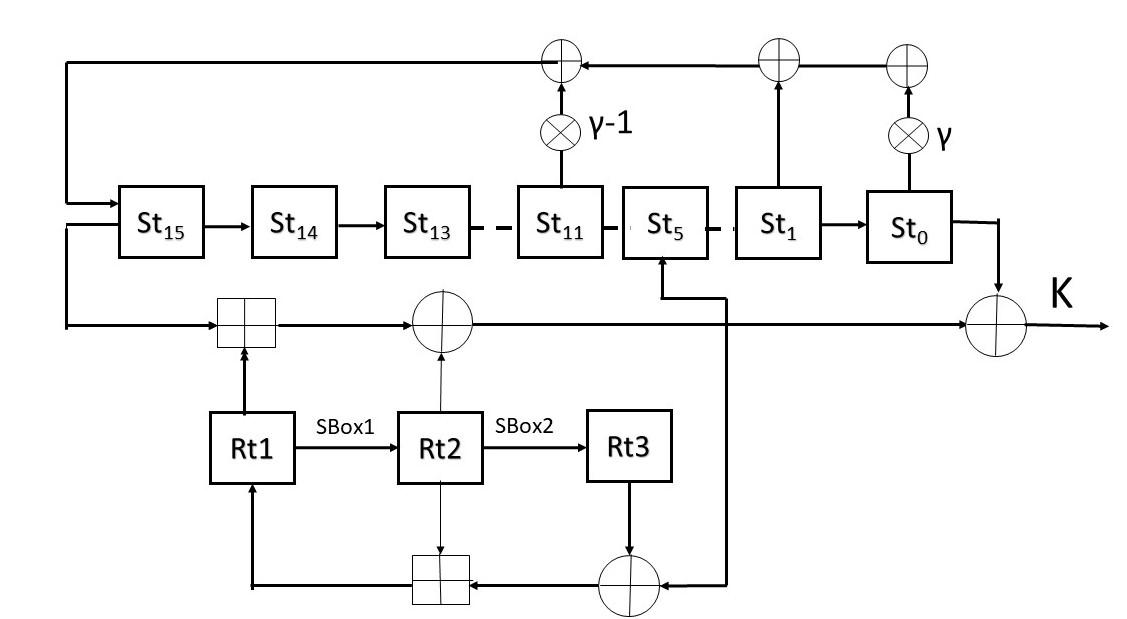}
	\caption{Building Block of SNOW 3G Cipher}
	\label{fig:snow-3g}
\end{figure}

  \subsection{Fast Correlation Attack on Word Based LFSRs}

  In a Fast Correlation Attack, 
  every window of the keystream $\{Y_1.....Y_N\}$ is seen as the noisy enncoding of the initial state of LFSR. The noise is used to model the effect of nonloinearity in the FSM of the stream cipher.    The parity check equations of the code are generated from the feedback equation of the LFSR along with the linear approximation of the FSM. Thus, the problem of recovering the initial state of the LFSR reduces to that of decoding an $[N,l]$ linear code where $l$ is the size of the LFSR. Using the  Wagners K-tree algorithm , this problem reduces to decoding a smaller code $[N_k,l']$ where $l'<l$. In (\cite{lee2008cryptanalysis,todo2018fast,gong2020fast,yang2019vectorized}), the feedback equation of degree $512$ is used  to successfully recover the state of the LFSR of SNOW 2.0 and SNOW 3G. On the contarary (\cite{zhang2015fast}) uses a feedback equation of degree $64$ over $F_{2^{8}}$.
\par
 Correlation attacks can be resisted by making the feedback equation dependent on the key. A procedure for generating a key dependent feedback equation is given in \cite{nandi2022key}. This can resist FCAs like the one given in (\cite{zhang2015fast}), where the feedback equation of the LFSR is seen as an equation with coefficients from extensions of $F_{2}$. However, in this scheme  the chararacteristic polynomial of the state transition matrix of the LFSR is the same as that in SNOW 2.0 and SNOW 3G. This polynomial can be used to generate a Linear Recrring Relation with coefficeints from $\mathbb{F}_2$ that the output of the LFSR  satisfies .This makes the scheme vulnerable to fast correlation attacks like the one given in \cite{gong2020fast} that employ the linear recurring relation over $\mathbb{F}_2$.

\subsection{$\sigma$-LFSRs and z-Primitive $\sigma$-LFSR}
 A $\sigma-$LFSR is a word based LFSR with multi-input multi-output delay blocks. The outputs of the delay blocks are multiplied by gain matrices and then added. This sum is fed back to the input of the first block. The output of a $\sigma$-LFSR with $m$-input $m$-output delay blocks is a vector sequence which satisfies the following recurrence relation,
\begin{equation}
\mathbf{s_{n+b}}  = \mathbf{s_{n+b-1}}B_{b-1} + \mathbf{s_{n+b-2}}B_{b-2} +\cdots + \mathbf{s_n}B_0
\end{equation}
where  $s_i \in \mathbf{F_2^{m}}$ and $B_i \in \mathbf{F_2^{m\times m}}$. The $B_i$s are the feedback gain matrices. The following matrix is called the configuration matrix of the $\sigma$-LFSR.
\begin{equation}\label{SU}
C=
\begin{bmatrix}
0 & I & 0 & \cdots & 0 \\
0 & 0 & I & \cdots & 0 \\
\vdots & \vdots & \vdots & \cdots &\vdots\\
0 & 0 & 0 & \cdots & I\\
B_0 & B_1 & B_2 & \cdots & B_{b-1}
\end{bmatrix} \in \mathbb{F}_2^{mb \times mb}
\end{equation}
where $0 \in\mathbb{F}_2^{m\times m}$ is the all zero matrix and $I\in F_2^{m \times m}$ is identity matrix. We shall refer to the structure of this matrix as the $m$-companion structure. The characteristic polynomial of a $\sigma$-LFSR is the characteristic polynomial of its configuration matrix. The period of the sequence generated by a $\sigma$-LFSR is maximum if its characteristic polynomial is primitive.

The vector got by appending the outputs of the delay blocks of a $\sigma$-LFSR at a given time instant is known as the \textbf{state vector of the $\sigma$-LFSR} at that time instant. This vector is an element of $\mathbb{F}_2^{mb}$. In this work, we will consider the state vector of a $\sigma$-LFSR to be a column vector. Thus, for $0\leq i \leq b-1$, if $s_i$ is the output of the $i$-th delay block at any given time instant, then the state vector of the $\sigma$-LFSR at that time instant is $c = \left(\begin{matrix} s_0\\s_1\\ \vdots \\s_{b-1}\end{matrix}\right)$

Two consecutive state vectors of a $\sigma$-LFSR are related by the following equation,
\begin{equation}
    c_{k+1} = C\times c_k
\end{equation}
where $c_k$ and $c_{k+1}$ are the state vectors at the $k$-th and $k+1$-th time instant respectively and $C$ is the configuration matrix of the $\sigma$-LFSR.

The output of a $\sigma$-LFSR with $b$, $m$-input $m$-output delay blocks can be seen as a collection of $m$ scalar sequences emanating from the $m$ outputs of the first delay block. These scalar sequences are known as the \textbf{component sequences} of the $\sigma$-LFSR. Any $n = mb$ consecutive entries of the component sequence constitute a state vector of the component sequences. \textbf{In this work, we consider state vectors of component sequences to be row vectors}. Two consecutive state vectors of any component sequence of a $\sigma$-LFSR, say $w_i$ and $w_{i+1}$,  are related by the following equation.
\begin{equation}
w_{i+1} = w_i \times M
\end{equation}
where the matrix $M$ is the companion matrix of the $\sigma$-LFSR i.e., if the characteristic polynomial of the $\sigma$-LFSR is $p(x) = x^n + c_{n-1}x^{n-1} + \cdots+ c_0$ then the matrix $M$ is given as follows,
\begin{equation}\label{P}
M=
\begin{bmatrix}
0 & 0 & 0 & \cdots & c_0 \\
1 & 0 & 0 & \cdots & c_1 \\
\vdots & \vdots & \vdots & \cdots &\vdots\\
0 & 0 & 0 & \cdots & c_{n-2}\\
0 & 0 & 0 & \cdots & c_{n-1}\\
\end{bmatrix}\in \mathbb{F}_2^{b \times b}.
\end{equation}

\begin{definition}
    The minimal polynomial of a sequence is the characteristic polynomial of the LFSR with the least number of delay blocks that generates the sequence.
\end{definition}

If a $\sigma$-LFSR is primitive, then each of its component sequences have the same period as that of the output sequence of the $\sigma$-LFSR. Further, the minimal polynomial of each of these sequences is the same as the characteristic polynomial of the $\sigma$-LFSR. In fact, each of these sequences are shifted versions of each other.

The distance between two component sequences $S^i$ and $S^j$ of a primitive $\sigma$-LFSR is defined as the number of left shifts needed to get to $S^j$ from $S^i$. Let $\delta$ denote the left shift operator i.e, for a sequence $S$, $\delta S(0) = S(1)$ and  $\delta^i S(0) = S(i)$.

\begin{definition}
The distance vector of a primitive $\sigma-$LFSR $S$ is defined as the set of integers $D=(d_1,d_2,\cdots,d_{m-1})$ where $d_i$ is the distance between the first and $i+1$-th component sequences, i.e. $S^i=\delta^{d_i} \times S^1$. 
\end{definition}

\begin{lemma}\label{tan}\cite{tan2012construction}
	Let $\alpha$ be a root of a primitive polynomial $f(x) = c_0 +c_1x +\hdots+c_{n-1}x^{n-1} +x^n$ over $GF(2)$ having degree $n=mb$. Let $D_m=(d_0,d_1,\cdots,d_{m-1})$ be the distance vector of a $\sigma$-LFSR with characteristic polynomial $f(x)$ having $b$, $m$-input $m$-output delay blocks. The set
	\[A=\{1,\alpha,\cdots,\alpha^{b-1},\alpha^{d_1},\alpha^{d_1+1},\cdots,\alpha^{d_1+b-1},\cdots,\alpha^{d_{m-1}},\alpha^{d_{m-1}+1},\cdots,\alpha^{d_{m-1}+b-1}\}\] 
	is a basis for $GF(2^{mb})$ as a vector space over $GF(2)$.
\end{lemma}

Now,  $\{1,\alpha,\alpha^2,\ldots,\alpha^{n-1}\}$ is a basis for $GF(2^{mb})$ as a vector space over $GF(2)$. With respect to this basis, for $1\leq i \leq n-1$, $\alpha^i$ is represented by the vector $e_1^i$ and every vector $v \in \mathbb{F}_2^n$ represents a polynomial in $\alpha$ with degree less that $n$ . Further, the linear map corresponding to multiplication by $\alpha$ is represented by the following matrix.  
\begin{equation}\label{P}
M_{\alpha}=
\begin{bmatrix}
0 & 1 & 0 & \cdots & 0 \\
0 & 0 & 1 & \cdots & 0 \\

\vdots & \vdots & \vdots & \cdots &\vdots\\
0 & 0 & 0 & \cdots & 1\\
c_0 & c_1 & c_2 & \cdots & c_{n-1-1}\\
\end{bmatrix}
\end{equation}

As there are no zero divisors in $GF(2^{mb})$, if each entry of the set $A$ in Lemma \ref{tan} is multiplied by a polynomial in $\alpha$ with degree less that $n$, the resulting set will also be a basis for $GF(2^{mb})$. Therefore, given any non-zero vector $v \in \mathbb{F}_2^n$ the set of vectors $\{v, vM_{\alpha},\ldots, vM_{\alpha}^{b-1}, vM_{\alpha}^{d_1}, vM_{\alpha}^{d_1 +1},\ldots, vM_{\alpha}^{d_1+b-1},\ldots,vM_{\alpha}^{d_m}, vM_{\alpha}^{d_m +1},\\ \ldots, vM_{\alpha}^{d_m+b-1} \}$ is a basis for $\mathbb{F}_2^{mb}$. Further, any matrix $M$ with characteristic polynomial $f(x)$ can be derived from $M_{\alpha}$ by a similarity transformation. Hence. given any non-zero vector $v$ and a matrix $M$ with characteristic polynomial $f(x)$, the set $\{v, vM,\ldots, vM^{b-1}, vM^{d_1}, vM^{d_1 +1},\ldots, vM^{d_1+b-1},\ldots,vM^{d_m}, vM^{d_m +1},\\ \ldots, vM^{d_m+b-1} \}$ is a basis for $\mathbb{F}_2^{mb}$. Given such a matrix $M$ and a distance vector $(d_1,d_2,\ldots,d_m$, one a generate a set of vectors $\{e_1^n,v_1 = e_1^nM^{d_1}, v_2 = e_1^nM^{d_2},\ldots,v_{m-1} = e_1^nM^{d_{m-1}}$. Such a set can be used to generate an $m$-companion matrix (and thereby a $\sigma$-LFSR configuration using the following lemma.

\begin{lemma} \label{Similarity}
\cite{krishnaswamy2012multisequences}
Let $M$ be the companion matrix of a given primitive polynomial $f(x)$ of degree $n$ where $n=mb$. Each m-companion matrix can be uniquely obtained from $M$ by means of a similarity transformation $P_1\times M\times P_1^{-1}$ where $P_1$ has the following structure:
	\begin{equation*}
	P_1=[e_1^n, v_1,\cdots,v_{m-1};e_1^nM, v_1M,\cdots,v_{m-1}M;\cdots, e_1^nM^{b-1}, v_1M^{b-1},\cdots,v_{m-1}M^{b-1} ]
	\end{equation*}
	where $M$ is the companion matrix for $f(x)$.
\end{lemma}
\begin{definition}
    \textbf{Matrix state of $\sigma-$LFSR sequence}
    Let $S=\{S(i)\in F_2^{m}\}_{i \in \mathbb{Z}}$ be a  sequence generated by a $\sigma$-LFSR with a primitive characteristic polynomial $f(x)$ of degree $n$. The $i$-th matrix state of $S$ is a matrix of $n$ consecutive elements of $S$ starting from the $i$-th element 
    \begin{equation*}
        Mat_S(i)=\{S(i),S(i+1),\cdots,S(i+b-1)\}_{m \times n}
    \end{equation*}
\end{definition}

    The dimension of a $\sigma-$LFSR sequence $S$ is the rank of any of its matrix states $Mat_S(i)$.
\begin{theorem}\label{Matrix_State}
Consider a $\sigma$-LFSR with $b$, $m$-input $m$-output delay blocks with configuration matrix $Q = P_1\times M\times P_1^{-1}$ where $P_1$ is as defined in Lemma \ref{Similarity} and $M$ is the companion matrix of the characteristic polynomial of the $\sigma$-LFSR. The first $m$ rows of the matrix $P_1$ constitute a matrix state of the sequence generated by the $\sigma$-LFSR.
\end{theorem}
\begin{proof}
For $1\leq i \leq mb$, let the $i$-th column of $P_1$ be denoted by $c_i$. Further, for some $1\leq i \leq mb$, let $c_i$ be a state vector of the $\sigma$-LFSR.  The next state vector of the $\sigma$-LFSR is $Q\times c_i$. This vector is computed as follows
\begin{equation}
   Q\times c_i =  (P_1 \times M \times P_1^{-1}) \times c_i=P_1 \times M \times (P_1^{-1} \times c_i)=P_1 \times M \times e_1 ^ {i}= P_1\times e_{1}^{i+1}=c_{i+1}
\end{equation}
Thus, the next state vector is the next column of $P_1$. Therefore, the $n$ columns of $P_1$ are $n$ consecutive state vectors of the $\sigma$-LFSR. Consequently, these vectors' first $m$ rows constitute $m$ consecutive outputs of the $\sigma$-LFSR. Hence, the first $m$ rows of the matrix $P_1$ are a matrix state of the sequence generated by the $\sigma$-LFSR.  \qed
\end{proof}
We now discuss a couple of results related to the binary field to motivate the concept of $z$-primitive $sigma$-LFSRs.
\begin{lemma}
	If $\alpha$ is a primitive element of $GF(2^{mb})$, $\alpha^z$ is  a primitive element of $GF(2^{m})$, where $z=\frac{2^{mb}-1}{2^m-1}$. 
\end{lemma}
\begin{proof}
As $\alpha$ is a primitive element of $GF(2^{mb})$, $\alpha$ generates the cyclic multiplicative group $F_{2^{mb}}^*$. Therefore, $\alpha^{2^{mb}-1}=1$. Further, the cardinality of the cyclic group generated by $\alpha^z$, i.e. the order of $\alpha^z$, is $\frac{2^{mb}-1}{GCD(2^{mb}-1,z)}$. Hence, if $z=\frac{2^n-1}{2^m-1}$, 
\begin{align}
order(\alpha^z)&=\frac{2^{mb}-1}{GCD(z,2^{mb}-1)}\\
                &= \frac{(2^{m}-1)(2^{m(b-1)}+2^{m(b-2)}+\cdots+1)}{(2^{m(b-1)}+2^{m(b-2)}+\cdots+1)}\\
                &=2^m-1
\end{align}
As, the order of the element $\alpha^z$ is $2^{m}-1$, it can be said that $\alpha^z$ is the generator of the subgroup $F_{2^{m}}^*$, i.e. $\alpha^z$ is  a primitive element of $GF(2^{m})$, where $z=\frac{2^{mb}-1}{2^m-1}$. \qed
\end{proof}
This gives rise to the following corollary
\begin{corollary}
	Let $f(x)$ be a primitive polynomial of degree $n=mb$ over $F_2$ and $\alpha$ be its root. If $z=\frac{2^{mb}-1}{2^m-1}$, a polynomial $g(x)$ having degree $m$ with root $\alpha^z$ is a primitive polynomial. Further, the polynomial $g(x)$ is given as follows:
	\begin{equation}
	g(x)=(x+\alpha^z)*(x+\alpha^{2z})*\cdots*(x+\alpha^{2^{m-1}z}) \pmod{f(\alpha)}
	\end{equation} 
\end{corollary}	
\begin{proof}
The primitiveness of $g(x)$ follows from Lemma 1. Since $mb$ is the smallest exponent of $alpha$ that equals 1, ($\alpha^z$, $\alpha^{2z}$, $\alpha^{2^2z}$,\ldots, $\alpha^{2^{m-1}z})$ are all distinct. Further, if $\alpha^z$ is a root of $g(x)$, so are $\alpha^{2z}$, $\alpha^{2^2z}$,\ldots, $\alpha^{2^{m-1}z}$. Therefore, $g(x)=(x+\alpha^z)*(x+\alpha^{2z})*\cdots*(x+\alpha^{2^{m-1}z}) \pmod{f(\alpha)}$. \qed
\end{proof}

\begin{definition}
	Let S be a primitive $\sigma$-LFSR with $b$, $m$-input $m$-output delay blocks. Let its distance vector be $D=(d_1,\cdots,d_{m-1})$. If all the elements of the distance vector are divisible by $z = \frac{2^{mb}-1}{2^m-1}$,i.e $z|d_i$ $\forall$ $ 1\leq i \leq m-1$, then S is called a $z$-primitive $\sigma$-LFSR.
\end{definition}

\begin{theorem}
\cite{tan2012construction}	The number of $z-$primitive $\sigma-$LFSR of having $b$, $m$-input $m$-output delay blocks is $\frac{|GL_{m}(GF(2))|}{2^m-1} \times\frac{\phi(2^{mb}-1)}{mb}  $, where ${|GL_{m}(GF(2))|}$ is the total $m \times m$ invertible matrices over $GF(2)$  and  $\frac{\phi(2^{mb}-1)}{mb}$ is the number of primitive polynomials of degree $mb$ over $GF(2)$.
\end{theorem}

\begin{remark}\label{bis1}
    Given $m$ and $b$, the set of $z$-primitive  $\sigma$-LFSR configurations is a subset of primitive $\sigma$-LFSR configurations. For the case $b=1$, both the sets have the same cardinality (This follows from Theorem 1 and Theorem 6.3.1 in \cite{krishnaswamy2012multisequences}). Therefore, when $b=1$, every $\sigma$-LFSR configuration is a $z$-primitive  $\sigma$-LFSR configuration.
\end{remark}

\begin{definition}
The $z$-set for a pair of integers $m$ and $b$ and a primitive polynomial $f$ having degree $mb$, denoted by $z_{mb}^f$, is the set of distance vectors of $z$-primitive $\sigma$-LFSRs having $b$, $m$-input $m$-output delay blocks and characteristic polynomial $f$.
\end{definition}

Note that there is a one to one correspondence between the set of distance vectors and the set of $\sigma$-LFSRs. Therefore, the cardinality of the $z$-set for a given pair of integers $m$ and $b$ and a given primitive polynomial $f$ is equal to the number of $z$-primitive LFSR configurations with $b$, $m$-input $m$-output delay blocks having characteristic polynomial $f$.

\begin{theorem}\label{bijection}
\cite{tan2012construction}	Let $\alpha$ and $\alpha^z$ be roots of primitive polynomials  $f(x)$ and $g(x)$ having degrees $mb$ and $m$ respectively. The following map $\phi$ from $z_{mb}^f$ and $z_m^g$is a bijection.
	\begin{align*}
		\phi\colon& z_{mb}^{f}\longrightarrow z_{m}^{g}\\
		&(d_0,d_1,\cdots,d_{m-1})\mapsto(\frac{d_0}{z},\frac{d_1}{z},\cdots,\frac{d_{m-1}}{z})
	\end{align*}
\end{theorem}
If $\alpha$ and $\alpha^z$ are roots of primitive polynomials  $f(x)$ and $g(x)$ having degrees $mb$ and $m$ respectively, the above theorem proves that there is one to one correspondence between $z$-primitive $\sigma$-LFSRs having $b$, $m$-input $m$-output delay blocks and characteristic polynomial $f(x)$ and $\sigma$-LFSR configurations with a single $m$-input $m$-output delay block and characteristic polynomial $g(x)$ (This is because every primitive $\sigma$-LFSR configuration is a $z$-primitive $\sigma$-LFSR configuration when $b=1$.)

\section{Proposed Scheme}
In this section, we describe a method of converting the existing feedback configuration of the LFSR into a random $z$-primitive $\sigma$-LFSR configuration. In contrast to the scheme given in \cite{nandi2022key}, the characteristic polynomial in this scheme is also randomly sampled from the set of primitive polynomials of a given degree. Information about this configuration is embedded in a public parameter. At the decryption side, the configuration is recovered from the public parameter using the secret key. In the proposed scheme, in addition to sharing the secret key and calculating the initial state of the LFSR, the following two steps have to be performed before the start of communication.
\begin{itemize}
    \item Generating a random $z$-primitive $\sigma$-LFSR configuration.
    \item Generating the public parameter using the state transition matrix of the $z$-primitive $\sigma$-LFSR and the secret key and declaring it to the entire network. 
\end{itemize}
 
\subsection{Generation of $z-$primitive $\sigma$ LFSR Configuration}
In this subsection, we propose a method to generate a random configuration matrix of $z$-primitive $\sigma-$LFSR with $b$, $m$-input $m$-output delay blocks. Here, we first randomly sample a primitive polynomial $f$ of degree $mb$. We then calculate a primitive polynomial $g$ with degree $m$ such that if $\alpha$ is the root of $f$ then $\alpha^z$ is the root of $g$ where $z = \frac{2^{mb}-1}{2^m-1}$. Thereafter, we  calculate the distance vector of a primitive $\sigma$-LFSR with a single $m$-input $m$-output delay block having characteristic polynomial $g$. This distance vector is used to generate an element of $z_{mb}^f$ . This is then used to generate a $\sigma$-LFSR configuration with characteristic polynomial $f$. Given a randomly sampled polynomial $f$ of degree $mb$, this process involves the following algorithms;
\begin{itemize}
\item An algorithm to find a primitive polynomial $g$ of degree $m$.
\item An algorithm to generate an element of $z_m^g$.
\item An algorithm to generate the $\sigma$-LFSR configuration with characteristic polynomial $f$ from an element of $z_m^g$.
\end{itemize}

\begin{algorithm}[H]\label{prim}
\caption{Algorithm to find a primitive polynomial of degree $m$ }
\textbf{Input:}
\begin{enumerate}
    \item  Companion matrix, $M_f$ of a randomly sampled primitive polynomial $f$ of degree $n$ over $GF(2)$,
\end{enumerate}

	\begin{algorithmic}[1]
		\Procedure{$Find\_primitive\_deg\_m$}{$M_f$}
        \State $z\gets \frac{2^n-1}{2^m-1}$
        \State $B\gets M_f^z$
        \State Sample a random vector $v\in F_2^n$.
        \State $i \gets 0$
         \While${ \quad i \ne (m-1)}$
        \State $v_1 \gets v\times B^i$
        \State $A[i,:]\gets v_1$    
        \State $i\gets i+1$
        \EndWhile
        \State Solve the following equation for $a=(a_0,a_1,\cdots, a_{m-1})\in F_2^m$.
        \[a.A=v\times B^m\]
		\EndProcedure 
\\  \textbf{Output:}
  \begin{enumerate}
  \item A primitive polynomial $g=a_0+a_1*x+a_2*x^2+\cdots+a_{m-1}*x^{m-1}+x^m$.
  
  \end{enumerate}
  
	\end{algorithmic}
	\label{UUU}
\end{algorithm}
In the above algorithm, the matrix $M_f$ is a representation of the root of $f$. The above algorithm calculates the minimal polynomial of the matrix $M_f^z$ which according to Lemma 1 is a primitive polynomial of degree $m$. \par
We now proceed to generating a random element of $z_m^g$. Observe that for a $\sigma$-LFSR with a  single delay block to have characteristic polynomial $g$ the feedback gain matrix can be any matrix with characteristic polynomial $g$. This is used to prove the following lemma leads to an algorithm that samples an element of $z_m^g$.

\begin{lemma}\label{fullrankmatrix}
Given a primitive polynomial $g$ of degree $m$, the set of matrix states of sequences generated by $\sigma$-LFSRs with a single delay block having characteristic polynomial $g$ is the set of all full rank matrices. Further, corresponding to each such $\sigma$-LFSR configuration there exists a unique matrix state with first row $e_1^m$.
\end{lemma}
\begin{proof}
Consider a $\sigma$-LFSR configuration with characteristic polynomial $g$ having a single delay block and feedback gain matrix $B$. The characteristic polynomial of $B$ is $g$. Given any vector $v$, the matrix $M_1 = [v,Bv,\ldots,B^{m-1}v]$ is a matrix state of this $\sigma$-LFSR configuration. As $g$ is a primitive polynomial, the matrix $M_1$ is invertible. To prove the first statement of the lemma, we will now prove that any invertible matrix having first column $v$ is a matrix state of a sequence generated by some $\sigma$-LFSR configuration with characteristic polynomial $g$ having a single delay block. 

Let $M_2$ be an arbitrary full rank matrix having first column $v$. Let $P = M_2M_1^{-1}$ and $B' = PBP^{-1}$. Comparing the first columns on both sides of the equation $PM_1 = M_2$, it is apparent that $v$ is an eigen vector of $P$ corresponding to the eigen value $1$. Now,
\begin{eqnarray*}
    M_2 = PM_1 &=& P[v,Bv,\ldots,B^{m-1}v]\\
    &=& [Pv,PBv,\ldots,PB^{m-1}v]\\
    &=& [v,PBP^{-1}Pv,\ldots,PB^{m-1}P^{-1}Pv]\\
    &=& [v,PBP^{-1}v,\ldots,PB^{m-1}P^{-1}v]\\
    &=& [v,B'v,\ldots,B'^{m-1}v]
\end{eqnarray*}
Thus, $M_2$ is a matrix state of a $\sigma$-LFSR with a feedback matrix $B'$ having characteristic polynomial $g$. Thus, any arbitrary matrix having first column $v$ is a matrix state of a $\sigma$-LFSR with characteristic polynomial $g$. Further, as $v$ can be arbitrarily chosen, any full rank matrix is a matrix state of a sequence generated by a $\sigma$-LFSR configuration with characteristic polynomial $g$ having a single delay block. Alternatively, the set of matrix states of sequences generated by $\sigma$-LFSRs with a single delay block having characteristic polynomial $g$ is the set of all full rank matrices.

Now, the minimal polynomial of each component sequence is $g(x)$. Therefore, every set of $m$-consecutive binary values occurs exactly once as a sub-sequence of these sequences in each period \cite{golomb2005signal}. Therefore, for each sequence generated by a $\sigma$-LFSR with characteristic polynomial $g$, there is a unique matrix state with the first row as $e_1^m$. As all sequences generated by a primitive $\sigma$-LFSR configuration are just shifted versions of each other, this matrix uniquely characterizes the $\sigma$-LFSR configuration.
\end{proof}
As a consequence of the above lemma, for a primitive polynomial $g(x)$ with degree $m$ and companion matrix $M_g$, given a random full rank matrix $M = [e_1^m;v_2;\ldots;v_{m-1}] \in \mathbb{F}_2^{m \times m}$, the set of integers $(d_1,d_2,\ldots,d_{m-1})$ such that $v_i = e_i^m M_g^i$ for $1 \leq i \leq m-1$ is an element of $z_m^g$. However, given an intertible matrix in $\mathbb{F}_2^{m \times m}$, calculating the set of integers $(d_1,d_2,\ldots,d_{m-1})$ is not  trivial. This is illustrated in the following lemma and corollary.
\begin{lemma}
    Let $a,b \in F_2^m$ and g be a primitive polynomial of degree $m$ over $GF(2)$. Let  $M_g\in F_2^{m \times m}$ be the companion matrix of  $g$. The calculation of the value $i$ in the equation $a\times M_g^i=b$ takes $\mathcal{O}(m\times e_K+e_K \times \sqrt{P_K}+m^3)$ time, where $2^m-1=\prod_{i=1}^{k}P_i^{e_i}$ and $P_k$  is the largest factor.
\end{lemma}
\begin{proof}
 $M_g$ is a representation of a root of $g$ and can therefore be seen as a primitive element of $F_{2^m}$. Therefore, $\mathcal{B}=\{M_g,M_g^2,\cdots,M_g^{m-1},M_g^m\}$ forms a basis for $F_2^m$. Hence, any  $M_g^i$ can be given by a linear combination of the elements in $\mathcal{B}$. Given an equation $a \times M_g^i=b$ for $a,b \in \mathbb{F}_2^m$, the value of $M_g^i$  can be computed by solving the following linear equation for $v$
 \begin{equation} \label{lineEq}
     v \times \begin{bmatrix}
         a\\a\times M_g\\a \times M_g^2\\ \vdots \\a \times M_g^{m-1}
     \end{bmatrix}_{m\times m}=b
 \end{equation}
 If $v = (a_0,a_1,\ldots,a_{m-1}$, then $M_g^i = (a_0\times I+a_1\times M_g+\cdots+a_{m-1}\times M_g^{m-1})$. Solving Equation \ref{lineEq} takes $\mathcal{O}(m^3)$ operations. The value of $i$ can then be found from $M_g^i$ by calculating the discrete logarithm using the Pohlig Hellman algorithm\cite{teske1999pohlig}. The time complexity of  Pohlig Hellman algorithm is $\mathcal{O}(m\times e_k+e_k\times \sqrt{P_k})$, where $2^m-1=\prod_{i=1}^{k}P_i^{e_i}$ and $P_k$  is the largest factor. Therefore the total time needed to calculate the value of $i$ is $\mathcal{O}(m\times e_K+e_K \times \sqrt{P_K}+m^3)$.
\end{proof}

Observe that calculating an element of $z_m^g$ from the matrix $M$ involves solving $(m-1)$ instances of the problem discussed in the above lemma. We therefore have the following corollary.

\begin{corollary}
    Given a matrix $M=(e_1^m,v_1,\cdots,v_{m-1})\in F_2^{m \times m}$ and a companion matrix $M_g\in F_2^{m \times m}$ of a primitive polynomial $g$, Finding the distance vectors $d_1,d_2,\cdots,d_{m-1}$ using Pohlig Hellman Algorithm from the equation $e_1^m\times M_g^{d_i}=v_i$ for $i\in[1,m-1]$ takes $\mathcal{O}((m-1)\times(m\times e_K+e_K \times \sqrt{P_K}+m^3)))$ time.
\end{corollary}
From Corollary 2 and Lemma 5, it is apparent that calculating an element of $z_m^g$ from a random invertible matrix is computationally expensive.
We therefore present a randomized algorithm (Algorithm \ref{UUU}) where an invertible matrix $M_1$ and the distance vector are simultaneously generated. The first row $M_1$ is taken as $e_1^m$. The algorithm runs for $m-1$ iterations. In the $i$-th iteration, an integer $d_i$  is randomly chosen from the set of integers ranging from 1 to $2^m-1$. The vector $e_1^mM_g$ is then evaluated. If this vector is linearly independent of the previously added rows of $M_1$, then $d_i$ is appended to the list of entries of the distance vector. Otherwise, the process is repeated with a new choice of $d_i$.  Thus, in each iteration a new entry is added to the distance vector and new row is appended to $M_1$. The linear independence of the newly added vector is checked by simultaneously generating a matrix $M_2$ whose entries are all 0 below the anti-diagonal. Further, for all $1\leq j \leq m$, the span of the first $j$ rows of $M_2$ is the same as the span of the first $j$ rows of $M_1$. The linear independence of the rows of $M_2$ ensures the linear independence of the rows of $M_1$.   The correctness of the algorithm is proved by Therem \ref{Correctnesszset}. Moreover, the proof gives an insight into the working of the algorithm.
  
\begin{algorithm}[H]
\caption{Randomized Algorithm to generate a z-set for $g$}
\textbf{Input:}
\begin{enumerate}
    \item  A companion matrix $M_g\in\mathbb{F}_2^{m \times m}$ of primitive polynomial $g$.
\end{enumerate}

	\begin{algorithmic}[1]
		\Procedure{$Z-Set$}{$M_g$}
		\State $M_1 \gets \begin{bmatrix}
		e_1^m\\
		z\\
		z\\
		\vdots\\
		z
		\end{bmatrix}\quad M_2\gets \begin{bmatrix}
		e_1^m\\
		z\\
		z\\
		\vdots\\
		z
		\end{bmatrix}\in 
        \mathbb{F}_2^{m \times m}$\Comment{z is zero vector of dimension m}
        \State $i \gets2$
         \While${\quad i \ne m}$
        \State $d \gets Random-Integer(2,2^m-1)$
        \State $ v \gets e_1^m \times M_g^d$
        \State $M_1[i,:]\gets v$
        \State $M_2[i,:]\gets v$
        \State $j \gets 1$
        \State \While${\quad j \ne i}$
        \State $M_2[i,:] \gets M_2[i,:]+M_2[i,m-j+1]*M_2[j,:]$
        \State $j\gets j+1$
         \EndWhile
        
            \State $\textbf{if}\quad M_2[i,m-i+1]==1$
            \State $\quad List.add(d)$
            \State ${\quad i\gets i+1}$
            \State $\textbf{endif}$
        
         \EndWhile
		\EndProcedure 
\\  \textbf{Output:}
  \begin{enumerate}
  \item Full rank Matrix $M_1\in F_2^{m \times m}$.
  \item The distance vector, List.
  
  \end{enumerate}
  
	\end{algorithmic}
	\label{UUU}
\end{algorithm}
\par
The main computational challange in the above algorithm is to find the value  $M_g^d$. This can be done using the Binary Exponentiation Algorithm(BEA)\cite{4565086} with matrix $M_g$ and integer $d$ as inputs. Using BEA,  for any $d \in [1,2^m-1]$, $M_g^d$ can be calculated in $\mathcal{O}(m^2\log_{2}(m))$ time. Here, $\mathcal{O}(m^2)$ time is required for multiplying two matrices of size $m \times m$ which is given in Algorithm \ref{MM}.
\begin{algorithm}[H]
\caption{Matrix-Matrix Multiplication over Binary Field}
\label{MM}
\begin{algorithmic}[1]
\Procedure{Matrix-Multiplication}{A,B}
\State \For{$i=0$ to $N$}
\State $Sum=0$
\State\For{$j=0$ to $N$}
\State  $Sum+=((\textbf{COUNTSETBITS}((A[i]\&\&B[j])))\&\& 1)\times 2^{N-1-j}$
\State \EndFor
\State $M[i]=Sum$
\State \EndFor
\State \EndProcedure
\end{algorithmic}
\end{algorithm}
The $COUNTSETBITS(x)$ function in the above algorithm counts the number of ones in $(A[i,:] \&\& B[:,j])$. The output of the bitwise AND of the COUNTSETBITS function with the integer 1 tells us the number is even or odd. It gives integer 1 when the count is odd and gives 0 when the count is even. It runs in $\mathcal O(1)$ time. The Algorithm \ref{CS} for COUNTSETBITS is given below:
\begin{algorithm}[H]
	\caption{Counting Set Bits in a Vector $\in F_2^{32}$}
	\label{CS}
	\begin{algorithmic}[1]
		\Procedure{COUNTSETBITS}{V}
		\State\For{$i=0$ to $255$}
		\State $V[i] \gets (i \&\& 1)+ Btable[\frac{i}{2}]$
		\State \EndFor
		\State $Result \gets(Btable[V \& 0xff]+Btable[(V>>8)\& 0xff]+ Btable[(V>>16)\& 0xff]+Btable[(V>>24)])\& 1$
		\State \EndProcedure
	\end{algorithmic}
\end{algorithm}

\par

\begin{theorem}\label{Correctnesszset}
    Given a primitive polynomial $g$ of degree $m$ over $GF(2)$, Algorithm 2 generates an element  of the $z_m^g$.
\end{theorem}
\begin{proof}
As a result of Lemma \ref{fullrankmatrix}, the above theorem stands proved if it is proved that the matrix $M_1$ is invertible. This is proved by inductively proving that each new row added to $M_1$ is linearly independent of all its previous rows. We now inductively prove the following statements;
\begin{itemize}
\item For all $1\leq i\leq m$, the first $i$ rows of $M_2$ are linear combinations of the first $i$ rows of $M_1$. Further, $M_2(i,m-i+1)=1$ and $M_2(i,k) = 0$ for all $k> m-i+1$.
\item  For all $1\leq i\leq m$, the first $i$ rows of $M_1$ are linearly independent.
\end{itemize}
As the first rows of $M_1$ and $M_2$ are $e_1^m$, both these statements are trivially true when $i=1$.
Assume that they are true for all $i<\ell<m$. 

Observe that, when $i=\ell$, due to lines 7 and 8 of Algorithm \ref{UUU}, $M_2[\ell,:]$ is initially equal to $M_1[\ell,:]$. $M_2[\ell,:]$ is then modified in the subsequent while loop (Line 11 to Line 15 of Algorithm \ref{UUU}) by adding it with some of the previous rows of $M_2$. By assumption, the first $\ell-1$ rows of $M_2$ are linear combinations of the first $\ell-1$ rows of $M_1$. Hence, at the end of this while loop, $M_2(\ell,:)$ is a linear combination of the first $\ell$ rows of $M_1$. Further, by assumption $M_2(j,m-j+1)=1$ and $M_2(j,k)=0$ for all $j<\ell$ and $k>m-j+1$. Therefore, when $i=\ell$ and $j = p<\ell$, in Line 12 of Algorithm \ref{UUU}, the $p$-th row of $M_2$ (whose $m-p+1$-th entry is 1) is added to $M(\ell,:)$ if and only if $M(i,m-p+1) = 1$. Therefore, after this addition $M(i,m-p+1)$ becomes zero. Further, all the entries of $M(\ell,:)$ that have been made zero in the previous iterations of the while loop remain zero as $M(p,k)=0$ for all $k>m-p+1$. Therefore, at the end of the inner while loop $M(\ell,k) = 0$ for all $k>m-\ell +1$. Now, the value of $i$ is incremented only when 
$M_2(\ell,m-\ell +1) = 1$. Therefore,  $M_2(\ell,m-\ell+1)=1$ and $M_2(\ell,k) = 0$ for all $k> m-\ell+1$.

The structure of the first $\ell$ rows of $M_2$ ensures that these rows are linearly independent. Further, since these rows are a linear combinations of the first $\ell$ rows of $M_1$, the first $\ell$ rows of $M_1$ are linearly independent.Thus both statements are true when $i=\ell$ and the theorem stands proved.\qed

\end{proof}
The element of $z_m^g$ generated in Algorithm-2 is then used to find an element of $z_{mb}^f$ using the bijection mentioned in Theorem 2. This is then used to generate the desired $\sigma$-LFSR configuration in the following algorithm.
\begin{algorithm}[H]
	\caption{Generation of M-companion matrix of Z primitive $\sigma-$ LFSR}
 \label{MAlgo}
	\begin{algorithmic}[1]
		\State \textbf{Input :} \begin{enumerate}
  \item Distance vector $D=(d_1,d_2,\cdots, d_{m-1})$(Generated from Algorithm-2). 
     \item Primitive polynomials $f$ of degree $n(mb)$ that was considered  in Algorithm 1. 
		\end{enumerate}

		\Procedure{\textbf{CONFIG-GEN}}{$INV_m,f,g$}      
.
		\State Compute $z \leftarrow\frac{2^{mb}-1}{2^m-1}$
		\State Construct the subspace $V_f$ of $F_2^{mb}$ by following:
		\[V_f=(e_1^n,e_1^n*M_f^{z*d_1},e_1^n*M_f^{z*d_2},\cdots,e_1^n*M_f^{z*d_{m-1}})=(e_1^n,w_1,\cdots,w_{m-1})\]\Comment{where $M_f$ is the companion matrix of the polynomial $f$}
	
		\State $Q \gets \begin{bmatrix} 
		e_1^n\\                        
		w_1\\
		\vdots\\
		w_{m-1}\\
		e_1^n\times M_f\\
		w_1 \times M_f\\
		\vdots\\
		w_{m-1} \times  M_f\\
		\vdots\\
		e_1^n\times \left( M_f\right)^{b-1}\\
		w_1 \times \left( M_f\right)^{b-1}\\
		\vdots\\
		w_{m} \times \left( M_f\right)^{b-1}
		\end{bmatrix}\in \mathbb{F}_2^{n \times n}$ 
		\State $A_{mb}\leftarrow Q \times M_f \times Q^{-1} $
		\EndProcedure
	\State \textbf{Output :} $M-$companion matrix $A_{mb}$ over $M_m(F_2)$. 	
	\end{algorithmic}
\end{algorithm}
\par

\begin{theorem}
Algorithm-\ref{MAlgo}  generates the configuration matrix of a $z-$primitive $\sigma-$LFSR whose characteristic polynomial is the primitive polynomial $f$ considered in Algorithm \ref{prim}. 
\end{theorem}
\begin{proof}
The set $D=\{d_1,\cdots,d_{m-1}\}$ is an element of $z_m^g$. Therefore, by Theorem 2 the set $D' = \{zd_1,\ldots,zd_{m-1}\}$ is an element of $z_{mb}^f$. Therefore, by the arguments given in Section 2, the vectors $(e_1^n,w_1,\ldots,w_{m-1}, e_1^nM_f,w_1M_f,\cdots,w_{m-1}M_f^{b-1},\\\ldots,e_1^nM_f^{b-1},w_1M_f^{b-1},\cdots,w_{m-1}M_f^{b-1})$ are linearly independent. Hence, by Lemma \ref{Similarity}, the matrix $QM_fQ^{-1}$ is an $m$-companion matrix with characteristic polynomial $f$. 

Now, by Theorem \ref{Matrix_State}, the first $m$-rows of the matrix $Q$ constitute a matrix state of the sequence generated by a $\sigma$-LFSR with configuration matrix $QM_fQ^{-1}$. Therefore, the distance vector of this $\sigma$-LFSR is $\{zd_1,\cdots,zd_{m-1}\}$. As each entry of this vector is a multiple of $z$, the $\sigma$-LFSR is $z$-primitive.

\end{proof}
\par
\textbf{Complexity Analysis:-}
Each $M_f^{z\times d_i}$ in Step 4 of Algorithm 4 can be calculated in $\mathcal{O}(n^2log_{2}(n))$ time. Since this operation is done $(n-1)$-times, the computational complexity of Step 4 is $\mathcal{O}(n^3)log(n)$. Further, $Q\times M_f\times Q^{-1}$ can be calculated in $\mathcal{O}(n^3)$ time. Thus, the overall time complexity of Algorithm 4 is $\mathcal{O}(n^3log_{2}(n))$.


\begin{example}
For the case where $m=3$ and $b=3$, we aim to generate a $z$-primitive $\sigma$-LFSR configuration with characteristic polynomial $f(x)=x^{9} + x^{6} + x^{4} + x^{3} + x^{2} + x + 1$. The following is the companion matrix of $f(x)$
\[M_f=\left(\begin{array}{rrrrrrrrr}
0 & 0 & 0 & 0 & 0 & 0 & 0 & 0 & 1 \\
1 & 0 & 0 & 0 & 0 & 0 & 0 & 0 & 1 \\
0 & 1 & 0 & 0 & 0 & 0 & 0 & 0 & 1 \\
0 & 0 & 1 & 0 & 0 & 0 & 0 & 0 & 1 \\
0 & 0 & 0 & 1 & 0 & 0 & 0 & 0 & 1 \\
0 & 0 & 0 & 0 & 1 & 0 & 0 & 0 & 0 \\
0 & 0 & 0 & 0 & 0 & 1 & 0 & 0 & 1 \\
0 & 0 & 0 & 0 & 0 & 0 & 1 & 0 & 0 \\
0 & 0 & 0 & 0 & 0 & 0 & 0 & 1 & 0
\end{array}\right)\]
 \begin{enumerate}
\setlength\itemsep{1em}
\item  $z=\frac{2^9-1}{2^3-1}=73$ 
\item $g(x)=(x+\alpha^{73})*(x+\alpha^{73*2})*(x+\alpha^{73*4})\pmod{f(x)}=x^3 + x^2 + 1$
, where $\alpha$ is root of $f(x)$.
\item Randomly sample a matrix from linear group $GL(3,GF(2))$. Let the sampled matrix be \[Mat=\left(\begin{array}{rr|r}
0 & 0 & 1 \\
\hline
 1 & 0 & 0 \\
0 & 1 & 0
\end{array}\right)\]

\item The distance vector corresponding to the above matrix is 
 $\{6,5\}$
\item The subspace $V_f$ in Step 4 of Algorithm 4 is as follows:
\[V_f=\{e_1^{9},e_1^{9}*M_f^{6},e_1^{9}*M_f^{5}\}\]

\item The matrix $Q$ in Step 5 of Algorithm 4 is as follows
\[Q=\left(\begin{array}{rrrrrrrrr}
0 & 0 & 0 & 0 & 0 & 0 & 0 & 0 & 1 \\
0 & 1 & 1 & 0 & 0 & 0 & 1 & 0 & 1 \\
1 & 0 & 0 & 0 & 0 & 1 & 0 & 1 & 1 \\
0 & 0 & 0 & 0 & 0 & 0 & 0 & 1 & 0 \\
1 & 1 & 0 & 0 & 0 & 1 & 0 & 1 & 1 \\
0 & 0 & 0 & 0 & 1 & 0 & 1 & 1 & 1 \\
0 & 0 & 0 & 0 & 0 & 0 & 1 & 0 & 0 \\
1 & 0 & 0 & 0 & 1 & 0 & 1 & 1 & 0 \\
0 & 0 & 0 & 1 & 0 & 1 & 1 & 1 & 0
\end{array}\right)\]
\item The configuration matrix of the $z$-primitive $\sigma$-LFSR is calculated as follows,
\[Q \times M_f \times Q^{-1}=\left(\begin{array}{rrrrrrrrr}
0 & 0 & 0 & 1 & 0 & 0 & 0 & 0 & 0 \\
0 & 0 & 0 & 0 & 1 & 0 & 0 & 0 & 0 \\
0 & 0 & 0 & 0 & 0 & 1 & 0 & 0 & 0 \\
0 & 0 & 0 & 0 & 0 & 0 & 1 & 0 & 0 \\
0 & 0 & 0 & 0 & 0 & 0 & 0 & 1 & 0 \\
0 & 0 & 0 & 0 & 0 & 0 & 0 & 0 & 1 \\
1 & 0 & 1 & 1 & 0 & 1 & 0 & 1 & 0 \\
1 & 0 & 0 & 1 & 0 & 0 & 0 & 0 & 1 \\
0 & 1 & 0 & 0 & 1 & 0 & 1 & 1 & 0
\end{array}\right)=\left(\begin{array}{rrr|rrr|rrr}
0 & 0 & 0 & 1 & 0 & 0 & 0 & 0 & 0 \\
0 & 0 & 0 & 0 & 1 & 0 & 0 & 0 & 0 \\
0 & 0 & 0 & 0 & 0 & 1 & 0 & 0 & 0 \\
\hline
 0 & 0 & 0 & 0 & 0 & 0 & 1 & 0 & 0 \\
0 & 0 & 0 & 0 & 0 & 0 & 0 & 1 & 0 \\
0 & 0 & 0 & 0 & 0 & 0 & 0 & 0 & 1 \\
\hline
 1 & 0 & 1 & 1 & 0 & 1 & 0 & 1 & 0 \\
1 & 0 & 0 & 1 & 0 & 0 & 0 & 0 & 1 \\
0 & 1 & 0 & 0 & 1 & 0 & 1 & 1 & 0
\end{array}\right)\]

\item The gain matrices for the configuration matrix ($Q \times M_f \times Q^{-1}$) are follows
$B_0$=
$\begin{bmatrix}
1 & 0 & 1 \\
1 & 0 & 0 \\
0 & 1 & 0 
\end{bmatrix}$\quad
$B_1$=
$\begin{bmatrix}
1 & 0 & 1 \\
1 & 0 & 0 \\
0 & 1 & 0 
\end{bmatrix}$ \quad
$B_2$=
$\begin{bmatrix}
0 & 1 & 0 \\
0 & 0 & 1 \\
1 & 1 & 0 
\end{bmatrix}$

\end{enumerate}
\end{example}

\subsection{Generation of the Public Parameter}

Recall that in the generation of the configuration matrix, the primitive polynomial in Algorithm \ref{prim} and the elements of the distance vector in Algorithm \ref{UUU} are randomly chosen independent of the key. Hence, the information needed to recover the $\sigma$-LFSR configuration is not completely contained in the key. Therefore, in addition to the key, we generate a publicly known parameter matrix $C$ to enable the reciever to recover the $\sigma$-LFSR configuration. 

Note that the proposed scheme has two secret keys $K_1, K_2\in F_2^{128}$ and two initialization vectors $IV_1,IV_2\in F_2^{128}$.  We now state the algorithm that generates the public matrix $C$ followed by a brief explanation.
\begin{algorithm}[H] \label{Public_Matrix}
	\caption{Generation of Public matrix on the Sender Side}
	\begin{algorithmic}[1]
		\State \textbf{Input: }
		\State Two random keys $K_1$, $K_2\in F_2^{128}$.
            \State Two public initialization vectors $IV_1$ and $IV_2\in F_2^{128}$.
		\Procedure{$\textbf{GEN\_CMatrix}$}{$f,K_1,K_2,IV_1,IV_2$}
 		\State Take the Configuration matrix $M_1$ from Algorithm-4 with m=32 and        b=16.
 	  \State Generate the vector $v_1\in F_2^{512}$  using the initial state generation algorithm of SNOW 2.0 with $(K_1,IV_1)$. Let $v_1 \leftarrow v_{1,1}||v_{1,2}||v_{1,3}||v_{1,4}$ where $v_{1,i}\in F_2^{128}$.
 		    \State   $U[1,:]\leftarrow v_1$.
 			\State \For{$i=2$ to $512$}
 		    \State  $v_{i,k}\leftarrow AES-128(v_{{i-1},k},K_2,IV_2)$\Comment{$k \in[4]$}
 			\State  $v_i\leftarrow v_{i,1}||v_{i,2}||v_{i,3}||v_{i,4}$
                \State  $U[i,:]\leftarrow v_i$
 			\State  \EndFor
 			\State  $a \leftarrow2^{512}-1$.
                \State \For{$i=0$ to $511$}
                \State $Mask[i]\leftarrow (a>>i)$
                \State \EndFor
                \State \For{$i=0$ to $511$}
                \State $U[i,:] \leftarrow (U[i] \&\& Mask[i])$
                \State $U[i,:]\leftarrow U[i,:] \oplus (1<<(511-i))$
                \State \EndFor
                \State $L \leftarrow U^{T}$
                \State $S\leftarrow L \times U$
 		    \State Compute $C=M_1[480:512,:] \times S$.
 		\EndProcedure
 		
 		\State Output: $C \in F_2^{32 \times 512}$.
	\end{algorithmic}
\end{algorithm}
In the above algorithm, the secret key $K_1,K_2$ is used to generate a secret matrix $S$ which is multiplied with the last $32$ rows of the M-companion matrix, $M_1$, to produce the matrix $C$. This $C$ is then made public. In order to generate a matrix $M$, a vector $v_1$ is generated using the (key,IV) pair $K_1,IV_1$. This is done using the procedure that is used to generate the initial state of SNOW 2.0 and SNOW 3G. $v_1$ is then divided into 4 equal parts of $128$ bits. Each of this part is encrypted using AES-128 algorithm using the key IV pair $K_2,IV_2$ pair. The encrypted words are then concatenated to generate the next vector $v_2$. This procedure is then followed recursively to generate 512 vectors. This vectors are considered as a rows of a matrix. The elements of the matrix below the diagonal are then discarded and and all the diagonal elements are made 1. This results an upper triangular matrix $U$. The matrix $S$ is then generated by multiplying the transpose of $U$ with $U$. That is $S=U \times L$. 

\begin{algorithm}[H]
	\caption{Recovery of Public matrix on the Receiver Side}
	\begin{algorithmic}[1]
		\State \textbf{Input: }
		\State Take two random keys $K_1$, $K_2\in F_2^{128}$.
            \State Take two public initialization vectors $IV_1$ and $IV_2\in F_2^{128}$.
		\Procedure{$\textbf{GEN\_CMatrix}$}{$f,K_1,K_2,IV_1,IV_2$}

 	  \State Generate the vector $v_1\in F_2^{512}$  using the initial state generation algorithm of SNOW 2.0 with $(K_1,IV_1)$. Let $v_1 \leftarrow v_{1,1}||v_{1,2}||v_{1,3}||v_{1,4}$ where $v_{1,i}\in F_2^{128}$.
 		    \State   $U[1,:]\leftarrow v_1$.
 			\State \For{$i=2$ to $512$}
 		    \State  $v_{i,k}\leftarrow AES-128(v_{{i-1},k},K_2,IV_2)$\Comment{$k \in[4]$}
 			\State  $v_i\leftarrow v_{i,1}||v_{i,2}||v_{i,3}||v_{i,4}$
                \State  $U[i,:]\leftarrow v_i$
 			\State  \EndFor
 			\State  $a \leftarrow2^{512}-1$.
                \State \For{$i=0$ to $511$}
                \State $Mask[i]\leftarrow (a>>i)$
                \State \EndFor
                \State \For{$i=0$ to $511$}
                \State $U[i,:] \leftarrow (U[i] \&\& Mask[i])$
                \State $U[i,:]\leftarrow U[i,:] \oplus (1<<(511-i))$
                \State \EndFor
                \State $L \leftarrow U.transpose()$
 		    \State Compute $S= U \times L$.
                \State Compute the inverse of $S$.
 		\EndProcedure
 		
 		\State Output: $C \times S^{-1} \in F_2^{32 \times 512}$.
	\end{algorithmic}
\end{algorithm}

 The receiv  computes the matrix $S$ from ($K_1, K_2$). The feedback gain matrices $B_i$ are recovered by multiplying the public matrix $C$(Generated in Algorithm-5) with the inverse of $S$.  Thus, the receiver regenerates the LFSR configuration from the keys ($K_1,K_2$) and the public matrix $C$\par 
Whenever a user intends to transmit confidential data using an LFSR-based word-oriented stream cipher, Algorithm-5 uses $(K_1, IV_1, K_2, IV_2)$ pair to mask the configuration matrix of the LFSR. The process of generating the challenge matrix occurs during the pre-computing phase prior to the initialization cycle of the Cipher. The $\sigma-$LFSR configuration corresponding to the M-companion matrix generated in Algorithm-4 is then used for Keystream Generation.

\subsubsection{Doing away with the Public Parameter $C$}
For the reciever to recover the $\sigma$-LFSR configuration, in addition to the shared key and publicly known IV, he/she should know the primitive polynomial that is sampled in Algorithm \ref{prim} and the entries of the distance vector. If these entries can be securely generated using the secret key, then the public parameter will not be required.
This will avoid the communication cost incurred for sharing the public parameter. 
\subsection{Initialization Phase: }
The initialization phase of this scheme is as same as the initialization phase of SNOW 2.0 or SNOW 3G. It runs for 32 clock cycles using the same feedback polynomial over $GF(2^{32})$ and the adversary is not allowed to access the keystreams during this period. At the last clock, it replaces the coefficients of the feedback polynomial over $GF(2^{32})$  by the gain matrices   $B_i\in F_2^{32 \times 32}$  $M_1$. 
\subsection{Keystream Generation Phase: }
In the Key generation phase, the LFSR part of SNOW 2.0 and SNOW 3G is regulated by the following equation:
\begin{equation}
	St_{15}^{t+1}=B_{0}St_{0}^{t} + B_{1}St_{1}^t + \cdots + B_{15}St_{15}^t   
\end{equation}
, where $B_i\in F_2^{32 \times 32}, i \in [16]$ is the gain matrices of $M_1$. And, each delay block($St_j,j\in[16]$) is updated as:
\begin{equation}\label{Up}
	St_{k}^{t+1}=\begin{cases}
		St_{k+t+1}^{0} &                                              0 \leq k+t+1 \leq 15\\
		\sum_{i=0}^{15}B_{i}St_{i}^{t+k} & k+t+1 >15
	\end{cases}
\end{equation}
FSM update part is the same as it was in SNOW 2.0 and SNOW 3G in section 4.

\subsection{Security of the proposed scheme}
Here, in addition to the keystream, the attacker has access to the public matrix $C$ generated by Algorithm \ref{Public_Matrix}. He/She could potentially use this matrix to retrieve the gain matrices of the LFSR. On average, a brute force attack will require $2^{255}$ guesses to get to the correct key.  We now show that other methods of deriving the feedback configuration are computationally more expensive.

\begin{lemma}
Given any symmetric invertible matrix $S' \in \mathbb{F}_2^{mb \times mb}$, there exists a matrix $M' \in \mathbb{F}_2^{m \times mb}$ such that the public parameter $C$ is a product of $M'$ and $S$.
\end{lemma}

\begin{proof}
Let $M_1$ be the configuration matrix of the $\sigma$-LFSR and let $C = M_1[mb-m+1:mb,:]\times S$ be the public parameter matrix derived in Algorithm \ref{Public_Matrix}. Given any symmetric invertible matrix $S' \in \mathbb{F}_2^{mb \times mb}$,
\begin{eqnarray*}
C &=& M_1[mb-m+1:mb,:]\times S \times S'^{-1}\times S'\\
&=& M'\times S'
\end{eqnarray*}
where $M' =  M_1[mb-m+1:mb,:]\times S \times S'^{-1} \in \mathbb{F}_2^{m \times mb}$.
\end{proof}
To find the gain matrices of the $\sigma$-LFSR from the public parameter, one can sample an invertible symmetric matrix $S'$  and find the corresponding values of $M'$ such that $C =  M'\times S'$. If the feedback configuration corresponding to $M'$ is $z$-primitive, then assuming this to be the feedback configuration, one can launch any existing attack to generate the initial state of the LFSR. The number of invertible symmetric matrices in 
$\mathbb{F}_2^{mb \times mb}$ is $2^{\frac{mb(mb-1)}{2}}$, Hence, the average number of attempts needed to get the correct invertible matrix is $\mathcal{O}(2^{\frac{mb(mb-1)}{2}-1})$ which is prohibitively high. 

Alternatively, one could sample a $z$-primitive LFSR configuration and consider the configuration matrix as a potential $M_1$. One can then check if the public parameter can be written as a product of the last $m$ rows of this matrix and a symmetric matrix $S$. This involves checking if a set of $m\times mb$ linear equations in $\frac{mb(mb-1)}{2}$ variables has a solution (For SNOW 2 and SNOW 3G the number of equations is 8192 and the number of variables is 130816). If this set of equations has a solution, one can launch any known attack to recover the initial state of the LFSR. For a given value of $m$ and $b$, the total number of $z$-primitive LFSR configurations is $\frac{|GL(m,\mathbb{F}_2)|}{2^m-1} \times \frac{\phi(2^{mb}-1)}{mb}$. For SNOW 2 and SNOW 3 this number turns out to be $2^{1493}$. Therefore, the average number of $z$-primitive LFSRs that  need to be sampled to get to the right configuration is approximately $2^{1492}$. 

\subsection{Resistance to Attacks: }
Several Known plaintext attacks like Algebraic Attacks\cite{billet2005resistance,courtois2008algebraic}, Distinguishing Attacks(\cite{watanabe2003distinguishing,nyberg2006improved}),Fast Correlation Attacks\cite{lee2008cryptanalysis,zhang2015fast,todo2018fast,yang2019vectorized,gong2020fast}, Guess and Determine Attacks\cite{New-SN-03,ahmadi2009heuristic}, Cache Timing Attacks(\cite{leander2009cache,brumley2010consecutive}) are reported for SNOW 2.0 and SNOW 3G. All these attacks either use the feedback equation of the $\sigma$-LFSR or the linear recurring relation corresponding to the characteristic polynomial of the  LFSR . A method of hiding the configuration matrix is explained in \cite{nandi2022key}. However, the characteristic polynomial of the $\sigma$-LFSR is publicly known. The scheme given in \cite{nandi2022key} is therefore vulnerable to schemes that use the characteristic polynomial of the $\sigma$-LFSR. This polynomial gives rise to a linear recurring equation of degree 512 with coefficients in $\mathbb{F}_2$.These schemes include the fast correlation attack given in \cite{gong2020fast} and the fault attack on SNOW3G given in \cite{New-SN-07}. As the characteristic polynomial in this scheme is not known. Therefore, to get to the characteristic polynomial the attacker has to keep sampling from the set of primitive polynomials of degree 512 till he/she gets to the correct characteristic polynomial. As the number of primitive polynomials of degree 512 over $\mathbb{F}_2$ is $2^{502}$, the attacker on average will sample $2^{501}$ polynomials. Alternatively, the attacker could try to generate the symmetric matrix $S$ in Algorithm \ref{Public_Matrix} by sampling the key space. This on average will take $2^{255}$ attempts to get to the correct key.

\begin{table}[]
    \centering
    \begin{tabular}{|c|c|c|c|}
    \hline
      \textbf{References} & \textbf{Applied Ciphers}  & \textbf{FCA Complexity} & \textbf{Our Scheme}\\
      \hline
      Lee et al.,2008\cite{lee2008cryptanalysis}           &   SNOW 2.0     & $2^{204.38}$   &   $2^{256}$  \\
      \hline
    Zhang et al.,2015\cite{zhang2015fast}         & SNOW 2.0   & $2^{164.15}$   &   $2^{256}$    \\
      \hline
    Todo et al.,2018\cite{funabiki2018several} &  SNOW 2.0     & $2^{162.91}$  & $2^{256}$\\
    \hline
    Yang et al.,2019\cite{yang2019vectorized} & SNOW 3G   &  $2^{177}$        & $2^{256}$\\
    \hline
    Gong et al.,2020\cite{gong2020fast}  &  SNOW 2.0, SNOW 3G   & $2^{162.86},2^{222.33}$   &$2^{256},2^{256}$\\
    \hline
    \end{tabular}
    \caption{Comparison results of our scheme with various schemes}
    \label{tab:my_label}
\end{table}

\section{Experiment Result}

The proposed scheme has been implemented in C (using a GCC compiler) on a machine with an Intel Core i5-1135G7 processor having  8GB RAM and a 512 GB HD drive. The parallel implementation of $e_1^n \times M_f^{z \times{d_i}}$ for $n=512$, $M_f\in F_2^{512 \times 512}$ and $1\leq |d_i|\leq 512$ where $i\in \{1,31\}$, took a total of $0.04$ seconds. The calculation  $P \times Q \times P^{-1}$(step-6 in Algorithm-\ref{MAlgo}) took another .08 seconds. Algorithm \ref{MAlgo} was completed in 0.13 seconds. The total initialization time for our scheme is on average 0.2 seconds (Averaged over 200 test cases). Besides, to accelerate the keystream generation process, $32$- bit vector-matrix multiplications in the $\sigma-$LFSR are done using Algorithm-\ref{MM} with constant time complexity, $\mathbb{O}(c)$, where $c=32$. This lead to an improvement in performance over existing implementations. The Key generation times(KGT) for SNOW 3G and our proposed scheme are given in the following table.
\begin{note}
     Look up table based implementation of LFSR in SNOW 3G can be cryptanalysed by cache timing attacks\cite{leander2009cache}. Hence, we have considered  the implementation of LFSR part of SNOW 3G(specially field multiplication over $GF(2^{32})$) is programmed without any look up tables.
\end{note}
\begin{table}[H]

    \centering
    \begin{tabular}{|c|c|c|}
    \hline
       \textbf{Number of Keystreams}  & \textbf{KGT for SNOW3G} & \textbf{KGT for Proposed SNOW3G)}\\
      \hline
         $2^{10}$     & .009490 Seconds  & .003138 Seconds \\
      \hline
       $2^{15}$ & .2032 Seconds & .0586 Seconds  \\
      \hline
       $2^{20}$ &  6.397 Seconds & 1.8152 Seconds\\
      \hline   
    \end{tabular}
\caption{Comparison of Key Generation Time of our proposed scheme with SNOW 3G}
    \label{tab:my_label}
\end{table}

\section{Conclusion and Future Work}
In this article, we have introduced a $z$-primitive sigma-LFSR generation algorithm to generate an $m$-companion matrix of m input-output and b number of delay blocks for word-based LFSR. Besides, to hide the feedback polynomial over $GF(2)$, we have multiplied a key-dependent invertible matrix with the $m$-companion matrix. Finding the part of the multiplied matrix that is shared as a public parameter is analogous to searching a symmetric matrix from the space of $2^{256 \times 511}$. Our scheme can resist Fast Correlation Attacks (FCA). We have shown that applying our scheme to SNOW 2.0, SNOW 3G, and Sosemanuk can withstand FCA. Besides, our scheme is robust against any known plaintext attacks based on the Feedback equation of the LFSR. The future works of this scheme are as follows
\begin{itemize}
\item  Implementation of  word based stream cipher with word size $m=64,128$ and comparison with SNOW V and SNOW VI\cite{ekdahl2021snow}.
\item Developing the scheme without the public parameter $C\in F_2^{32 \times 512}$. 
\end{itemize}

\bibliographystyle{splncs04}
\bibliography{zsigma}
\end{document}